\documentclass{article}
\usepackage[english]{babel}
\usepackage{geometry}
\usepackage[utf8]{inputenc}  
\usepackage[T1]{fontenc}     
\usepackage{amsmath, amssymb, amsthm} 
\usepackage{graphicx}        
\usepackage{hyperref}        
\usepackage{newtxtext,newtxmath}  
\usepackage[plain,noend]{algorithm2e}  
\usepackage[numbers]{natbib}

\newtheorem{theorem}{Theorem}[section]
\newtheorem{lemma}[theorem]{Lemma}

\newtheorem{corollary}[theorem]{Corollary}
\newtheorem{proposition}[theorem]{Proposition}  
\newtheorem{assumption}[theorem]{Assumption}  
\title{Log-Gaussian Cox Processes on General Metric Graphs}

\author{
    David Bolin$^{1}$\thanks{Authors in alphabetical order.} \quad
    Damilya Saduakhas$^{1}$ \quad
    Alexandre B. Simas$^{1}$ \\[0.5em]
    \small $^{1}$Statistics Program, King Abdullah University of Science and Technology (KAUST), \\ 
    \small Thuwal 23955, Saudi Arabia \\[0.5em]
    \small \texttt{david.bolin@kaust.edu.sa}, 
    \texttt{damilya.saduakhas@kaust.edu.sa}, 
    \texttt{alexandre.simas@kaust.edu.sa}
}

\date{}

\begin{document}
\maketitle

\begin{abstract}
The modeling of spatial point processes has advanced considerably, yet extending these models to non-Euclidean domains, such as road networks, remains a challenging problem. We propose a novel framework for log-Gaussian Cox processes on general compact metric graphs by leveraging the Gaussian Whittle--Mat\'ern fields, which are solutions to fractional-order stochastic differential equations on metric graphs. To achieve computationally efficient likelihood-based inference, we introduce a numerical approximation of the likelihood that eliminates the need to approximate the Gaussian process. This method, coupled with the exact evaluation of finite-dimensional distributions for Whittle--Mat\'ern  fields with integer smoothness, ensures scalability and theoretical rigour, with derived convergence rates for posterior distributions. The framework is implemented in the open-source MetricGraph R package, which integrates seamlessly with R-INLA to support fully Bayesian inference. We demonstrate the applicability and scalability of this approach through an analysis of road accident data from Al-Ahsa, Saudi Arabia, consisting of over 150,000 road segments. By identifying high-risk road segments using exceedance probabilities and excursion sets, our framework provides localized insights into accident hotspots and offers a powerful tool for modeling spatial point processes directly on complex networks.
\end{abstract}

\vspace{1em}
\noindent \textbf{Keywords:} Gaussian random field; Linear network; Log-Gaussian Cox process; Metric Graph; Spatial point process; Stochastic Partial Differential Equation.

\section{Introduction}\label{sec:intro}

Spatial point processes are indispensable for analyzing event patterns across diverse fields, from epidemiology to urban safety planning \citep{LGCP8}. Among these, log-Gaussian Cox processes \citep{LGCP1} have become a widely adopted framework for modeling spatially varying intensities, leveraging latent Gaussian fields to flexibly incorporate spatial covariates \citep{LGCP9}. Formally, a log-Gaussian Cox process $X$ on a spatial domain $D$ is defined as a Cox process with an intensity function 
\begin{equation}\label{eq:lgcp_space}
    \Lambda(s) = \exp(u(s)), \quad s \in D,
\end{equation}
where $u$ is a Gaussian process. That is, for any Borel set $B\in\mathcal{B}(D)$, the cardinality of the set $X\cap B$,  conditioned on a realization $\lambda$ of $\Lambda$, is a Poisson distributed random variable with intensity $\int_B \lambda(s) ds$. 

Although log-Gaussian Cox processes are well-established in Euclidean domains \citep{LGCP7}, their extension to network-constrained geometries—such as road networks—remains both theoretically and computationally challenging. Traditional models often rely on Euclidean metrics, which misrepresent connectivity and spatial dependencies in networks where events are restricted to linear pathways \citep{GRAPH7,ROAD5}.  
Recent work on linear network point processes has addressed these limitations by adopting network-based distances such as shortest-path or resistance metric, accommodating nonstationarity, and developing tailored summary statistics for inhomogeneous processes \citep{GRAPH12,ROAD8}. A growing body of research has extended Cox and log-Gaussian Cox process models to incorporate the structure of road networks, with applications ranging from traffic accident analysis to urban mobility \citep{LGCP10, LGCP5}. A key advance by \citet{GRAPH7} introduced isotropic covariance functions on linear networks using the resistance metric, replacing Euclidean distances to better align with network connectivity. Building on this work, \citet{LGCP5} developed Cox processes for linear networks based on these isotropic covariance functions. Those models could easily be extended to metric graphs with Euclidean edges, which is a more general class of metric graphs than linear networks; however, having Euclidean edges is still a restriction that often is not satisfied for large graphs such as the one we will later analyze in this work. Further, 
as noted by \citet{ACCIDENT4}, it is not clear that isotropy is a desirable feature as it may fail to capture anisotropic dependencies inherent to real-world road networks, underscoring the need for more flexible models \citep{ROAD8}.

Alternative approaches embed networks into the Euclidean plane \citep{ROAD3}, approximating spatial processes in \(\mathbb{R}^2\). This embedding, however, leads to significant inaccuracies, as such models cannot capture the constraints imposed by network geometry—such as the fact that events occur only along roads and not in surrounding spaces—resulting in uncontrolled boundary effects. 
Further, large-scale applications, such as city-wide accident modeling, result in additional computational bottlenecks. While the spatstat.linnet package \citep{ACCIDENT5} provides tools for exploratory analysis on linear networks, it supports only Poisson processes and lacks inference capabilities for log-Gaussian Cox processes. The coxln package \citep{LGCP5}, designed for network-based Cox processes, implements the log-Gaussian Cox processes based on isotropic covariance functions. This package is thus limited to problems on metric graphs with Euclidean edges. Additionally, based on our tests, coxln does not scale well computationally to large graphs and offers limited integration of spatial covariates, which often are critical for real-world applications.

In this work, in Section~\ref{sec:section2}, we for the first time propose a class of log-Gaussian Cox processes that are well-defined on any compact metric graph. 
These are thus not only well-defined on linear networks and metric graphs with Euclidean edges, but on arbitrary compact metric graphs which can accommodate complex geometries common in applications. The model class is based on the recently introduced Gaussian Whittle--Mat\'ern fields \citep{GRAPH4,GRAPH1} on metric graphs, which can be viewed as an extension of the popular stochastic partial differential equation approach to metric graphs \citep{GRAPH1,SOFTWARE10}. Besides being well-defined on general metric graphs, another advantage of these processes is that they can be made differentiable to allow for smoothly varying intensities, which is so far not possible through the approach based on isotropic fields, even if the graph has Euclidean edges. 

Methodologically, in Section~\ref{sec:section3}, we propose a version of the computationally efficient likelihood-based inference algorithm proposed by \citet{LGCP3}, adapted to the metric graph setting, which relies on a quadrature approximation of the integral appearing in the likelihood. However, contrary to \citet{LGCP3}, our approach does not rely on any approximation of the Gaussian process. Besides being more accurate, this also allows us to give a complete theoretical justification of the method by deriving convergence rates for the resulting posterior distributions. Finally, we combine this approximation with the fact that the Whittle--Mat\'ern fields have Markov properties for integer smoothness parameters \citep{GRAPH9} to obtain a theoretically justified and computationally efficient inference method that scales well to applications involving large datasets and networks. 

We demonstrate the utility of our framework through an analysis of traffic accidents in Saudi Arabia’s Al-Ahsa governorate in Section~\ref{sec:application}. Using high-resolution road network data comprising 178,801 segments, we generate intensity maps and identify high-risk regions by estimating excursion sets of the latent Gaussian process in Section~\ref{sec:excursions}. The open-source MetricGraph R package \citep{SOFTWARE3} implements our methodology, providing a scalable interface to the R-INLA software for fitting log-Gaussian Cox processes in a fully Bayesian framework \citep{SOFTWARE10,SOFTWARE12}. The article concludes with a discussion in the final Section~\ref{sec:discussion}. 

\section{Whittle--Mat\'ern Log-Gaussian Cox Processes on Metric Graphs}\label{sec:section2}



\subsection{Gaussian Whittle--Mat\'ern Fields}\label{sec:WM_graph}

Consider a compact, undirected metric graph \( \Gamma \), consisting of a finite set of vertices \( \mathcal{V} = \{v_i\} \) and a finite set of edges \( \mathcal{E} = \{e_j\} \), where each edge \( e \) is a curve of finite length \( \ell_e \in (0, \infty) \) connecting two vertices. Each edge \( e \) is defined by a pair of vertices \( (\bar{e}, \underline{e}) = (v_i, v_k) \). For each vertex \( v \in \mathcal{V} \), denote the set of edges incident to \( v \) by \( \mathcal{E}_v \), with degree \( \deg(v) = |\mathcal{E}_v| \). A point \( s \in \Gamma \) is located on an edge \( e \) and is represented as \( s = (e, t) \), where \( t \in [0, \ell_e] \). Because the lengths $\ell_e $ are finite, the graph is compact. We assume the graph is connected, ensuring a path exists between any two vertices, and equip the graph with the geodesic distance $d(\cdot, \cdot)$, which measures the shortest path length between any two points in $\Gamma$. This compact and connected structure facilitates defining statistical models directly on the graph edges, enabling a natural extension of spatial processes to network-constrained domains and going beyond traditional combinatorial graphs. Specifically, a real-valued function \( f \) on \( \Gamma \) is a collection of real-valued functions \( \{f_e\}_{e \in \mathcal{E}} \), where \( f_e \colon [0, \ell_e] \to \mathbb{R} \) for each $e \in \mathcal{E}$. 

A Whittle--Mat\'ern field $u$ is defined as a solution to the following fractional-order differential equation on $\Gamma$:
\begin{equation}\label{eq:WM_graph}
\left(\kappa^2 - \Delta_{\Gamma}\right)^{\alpha / 2} (\tau u) = \mathcal{W} \quad \text{on } \Gamma,
\end{equation}
where $\kappa, \tau > 0$ control the practical correlation range and the marginal variance of $u$, $\alpha = \nu + 1/2 > 0$ determines the smoothness, $\mathcal{W}$ is Gaussian white noise on a probability space $(\Omega, \mathcal{F}, \mathbb{P})$, and $\Delta_{\Gamma}$ is the Kirchhoff--Laplacian. 

Specifically, $\Delta_{\Gamma}$ acts as the second derivative on each edge, with domain 
$\mathcal{D}(\Delta_{\Gamma}) = \widetilde{H}^2(\Gamma) \cap C(\Gamma) \cap K(\Gamma)$. Here, $\widetilde{H}^k(\Gamma) = \bigoplus_{e \in \mathcal{E}} H^k(e)$ is the space of all functions which are $k$ times weakly differentiable on the edges, where $H^k(e)$ denotes the standard Sobolev space of order $k$ on the edge $e$.  Furthermore, the space $C(\Gamma) = \left\{ f \in L_2(\Gamma) : f \text{ is continuous} \right\}$ is the class of continuous functions on $\Gamma$, and \( L_2(\Gamma) = \bigoplus_{e \in \mathcal{E}} L_2(e) \) is the space of square-integrable functions on $\Gamma$, equipped with the norm $\|f\|_{L_2(\Gamma)}^2 = \sum_{e \in \mathcal{E}} \|f_e\|_{L_2(e)}^2$.
Finally, 
\begin{equation}\label{eq:Kirchoff}
K(\Gamma) = \left\{ f \in \widetilde{H}^2(\Gamma)\cap C(\Gamma) : \sum_{e \in \mathcal{E}_v} \partial_e f(v) = 0 \right\},
\end{equation}
is the space of functions satisfying the so-called Kirchhoff vertex conditions, which are continuous functions satisfying the vertex condition $\sum_{e \in \mathcal{E}_v} \partial_e f(v) = 0$, where $\partial_e$ denotes the directional derivative away from the vertex. This imposes a sum-to-zero constraint on the derivatives at the vertices, meaning that the derivative is continuous at vertices of degree 2.

The fractional power in \eqref{eq:WM_graph} is then defined in the spectral sense, and the Gaussian white noise $\mathcal{W}$ on $L_2(\Gamma)$ can be represented through a spectral decomposition involving the eigenfunctions $\{\varphi_i\}_{i \in \mathbb{N}}$ of  $\Delta_{\Gamma}$:
\begin{equation}
\mathcal{W} = \sum_{i \in \mathbb{N}} \xi_i \varphi_i, \quad \xi_i \stackrel{\text{iid}}{\sim} N(0, 1) \text{ on } (\Omega, \mathcal{F}, \mathbb{P}).
\end{equation}
One can show that, for $\alpha > 1/2$, \eqref{eq:WM_graph} has a unique solution in this setting \citep{GRAPH1}, which is a centered Gaussian process $u \in L_2(\Gamma)$ ($\mathbb{P}$-almost surely) and is referred to as a Whittle--Mat\'ern field. See \citet{GRAPH1} for details. 

One notable feature of the Whittle--Mat\'ern fields is that the parameter $\alpha$ directly controls the sample path regularity of the process. Specifically, if $\alpha>1/2$, there exists a modification of $u$ with $\gamma$-Hölder continuous sample paths for any $0 < \gamma < \min\{\alpha - 1/2, 1/2\}$. Moreover, if $\alpha > 3/2$, then $u\in H^1(\Gamma) = \widetilde{H}^1(\Gamma)\cap C(\Gamma)$ $\mathbb{P}$-almost surely \citep{GRAPH1}. 

\subsection{Log-Gaussian Cox processes based on Whittle--Mat\'ern fields}
A log-Gaussian Cox process on a metric graph $\Gamma$ is defined as a Cox process on $\Gamma$ with a random intensity function $\Lambda(s) = \exp(x(s))$ for $s\in\Gamma$, where $x$ is a Gaussian process on $\Gamma$. 
We introduce the Whittle--Mat\'ern log-Gaussian Cox process as a special case where $x = m + u$, for some deterministic function $m$ and a centered Gaussian Whittle--Mat\'ern field $u$. As an immediate consequence of the well-posedness of \eqref{eq:WM_graph}, this defines a well-defined log-Gaussian Cox process.

\begin{proposition}
Let \(u=\{u(s): s \in \Gamma\}\) be a Gaussian Whittle-Matérn field with $\alpha > 1/2$ on a compact metric graph $\Gamma$ and $m\in L_2(\Gamma)$. Then $\Lambda = \exp(m + u)$ is a valid stochastic intensity function for a log-Gaussian Cox process on $\Gamma$.     
\end{proposition}

The following proposition is easily verified along similar lines as Proposition 5.4 in \citet{LGCP7} using the expression for the Laplace transform of a normally distributed random variable.

\begin{proposition} 
A Whittle--Mat\'ern log-Gaussian Cox process  \(X\) has an intensity function and a pair correlation function given by
\begin{equation}
\rho(s)=\exp(m(s)+c_{\kappa,\tau,\alpha}(s, s)/2), \quad g(s, t)=\exp(c_{\kappa,\tau,\alpha}(s, t))
\end{equation}
where \(c_{\kappa,\tau,\alpha}\) is the covariance function of the underlying Whittle--Mat\'ern field with parameters $\kappa,\tau >0$ and $\alpha > 1/2$. 
\end{proposition}

Notably, the Whittle--Mat\'ern fields are in general not isotropic, which means that the corresponding log-Gaussian Cox processes are not isotropic either, even if $m = 0$. Specifically, the variance function $v_{\kappa,\tau,\alpha}(s) = c_{\kappa,\tau,\alpha}(s,s)$ of a Whittle--Mat\'ern field is not constant, where $v(s)$ is lower near vertices with degree greater than 2 and higher near vertices of degree 1. An illustration of the intensity function and the pair correlation function for a Whittle--Mat\'ern field can be seen in Figure~\ref{fig:intensity}.

The non-isotropy of the field is often realistic for applications on metric graphs \citep{ACCIDENT7}. However, for the modeling of traffic accidents, the specific form of non-isotropy may not be desirable since this implies that the intensity is slightly lower close to intersections. We can address this issue by replacing the Whittle--Mat\'ern field $u$ by a variance-stationary Whittle--Mat\'ern field $\hat{u}$ \citep{GRAPH11}, defined through the equation
$$
(\kappa^2 - \Delta_{\Gamma})^{\alpha/2} (\sigma_\kappa \hat{u}) = \sigma \mathcal{W},
$$
where $\sigma_\kappa$ represents the marginal standard deviations of a Whittle--Mat\'ern field with $\tau = 1$ and $\sigma>0$ is a parameter that is the standard-deviation of $\hat{u}$. Thus, using a variance-stationary Whittle--Mat\'ern field, we obtain a corresponding log-Gaussian Cox process with intensity function $\rho(s) = \exp(m(s) + \sigma^2/2)$ and pair correlation function
$g(s, t)=\exp(\sigma^2 c_{\kappa,1,\alpha}(s, t)/(\sigma_{\kappa}(s)\sigma_{\kappa}(t)))$.
This implies that setting $m = 0$  results in a log-Gaussian Cox process with a constant intensity function. 

\begin{figure}
    \centering
    \includegraphics[width=0.49\linewidth]{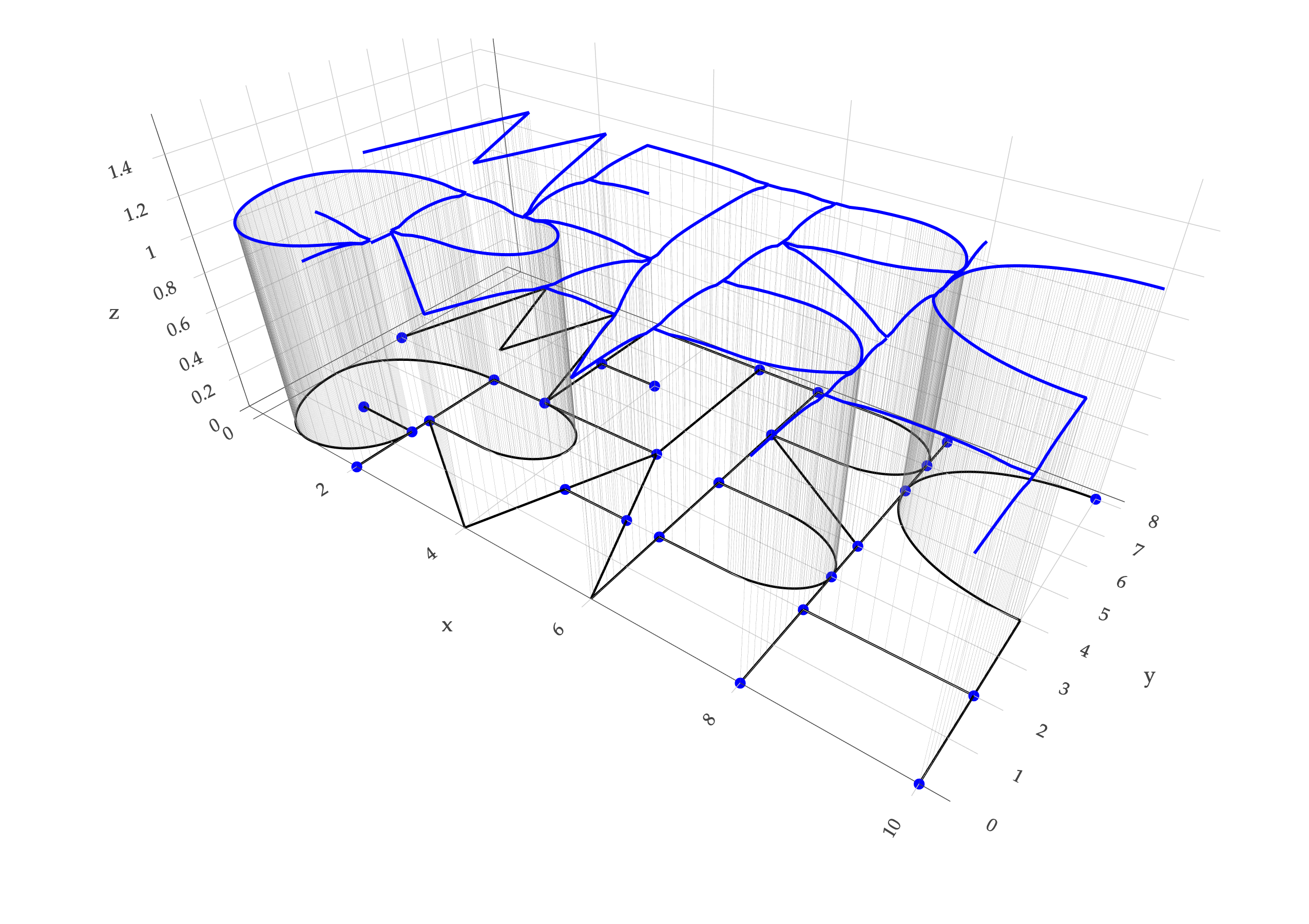}
    \includegraphics[width=0.49\linewidth]{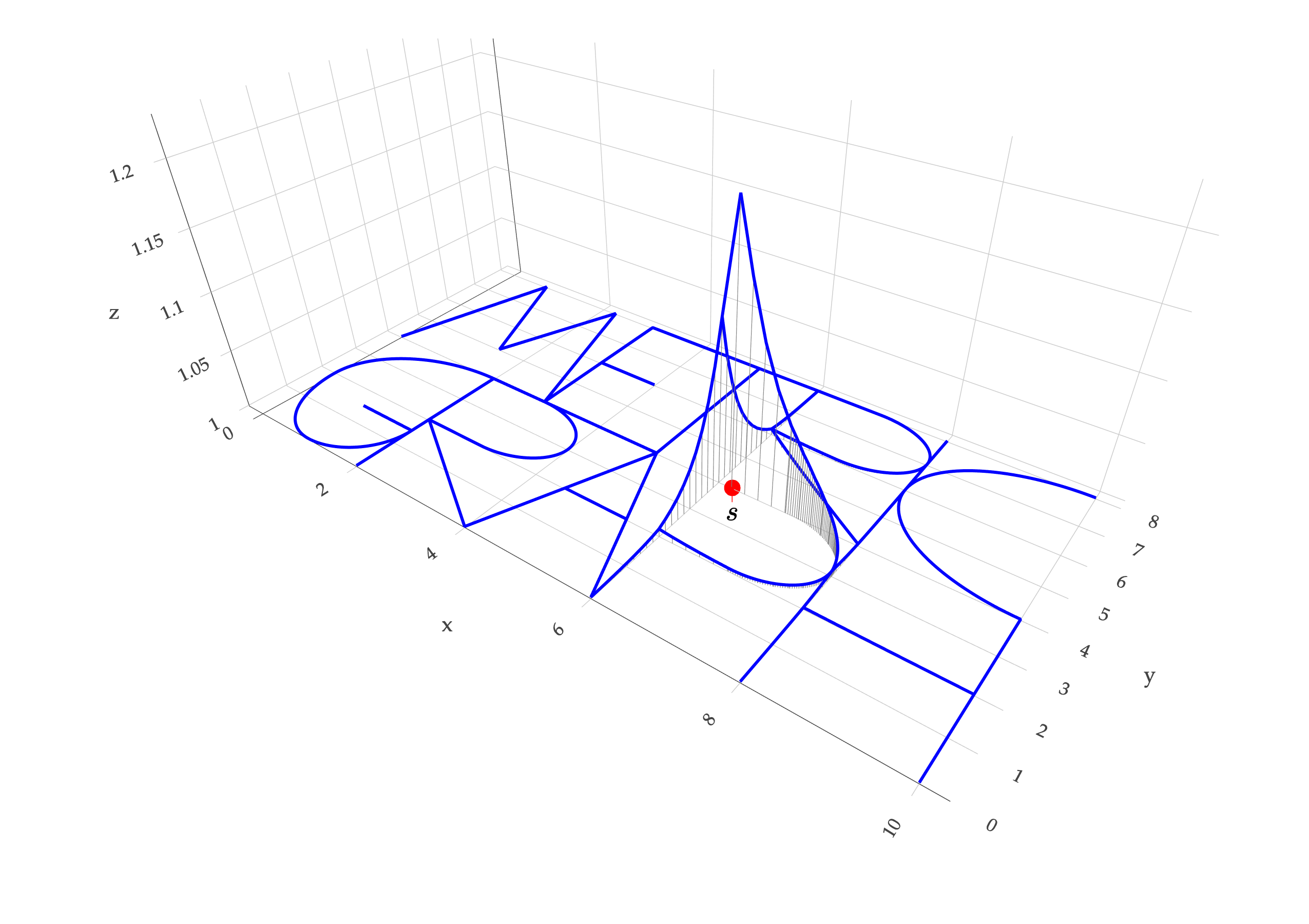} 
    \caption{Intensity function $\rho$ (left) and pair correlation function $g(s,\cdot)$ (right) for a Whittle--Mat\'ern log-Gaussian Cox process with $\alpha = 1$, $\kappa = 2$, $\tau = 1$. The red point in the right panel indicates the location $s$.}
    \label{fig:intensity}
\end{figure}

\section{Statistical Inference}\label{sec:section3}
Likelihood-based inference for log-Gaussian Cox processes is challenging as the likelihood of the data conditionally on the intensity function is 
\begin{equation}\label{eq:lgcp_likelihood}
    \pi(Y \mid \lambda)=\exp \left\{|\Gamma|-\int_{\Gamma} \lambda(s) \mathrm{d} s\right\} \prod_{s_i \in Y} \lambda\left(s_i\right).
\end{equation}
The first term represents the integral of the intensity function over the graph, which accounts for the expected number of points in the domain, while the second term evaluates the intensity at the observed point locations. Due to the complexity of the intensity function, this integral is typically intractable and must be approximated numerically.
One common way to approximate the likelihood for Euclidean domains is to partition the domain into a grid and transform the likelihood into a standard Poisson regression likelihood \citep{Illian2012}. Another approach suggested by \citet{LGCP3} for SPDE-based models is to combine a finite element approximation of the Gaussian process with a numerical approximation of the integral.

In the metric graph setting, we propose a third option, which is to approximate the integral using a mid-point rule, without approximating the field itself. The reason for this is that, as we will see below, we can then provide theoretical guarantees for the full approximation of the corresponding posterior distribution of the latent Gaussian process. 
This idea results in an approximation of the log-likelihood of the form 
\begin{equation}
    \log \pi(Y \mid \lambda)=|\Gamma|- \sum_{i=1}^p \widetilde{a}_i \lambda(\widetilde{s}_i) +  \sum_{i=1}^N \log \lambda(s_i),
\end{equation}
where \( \widetilde{s}_i \) are the midpoints of a mesh over the graph  and \( \widetilde{a} = (\widetilde{a}_1,\ldots, \widetilde{a}_p) \) are corresponding integration weights. That is, $\widetilde{a}_i$ is the area corresponding to the $i$th node in the integration mesh. 
Let $0_{n}$ and $1_n$ denote row vectors of length $n$ with all elements equal to $0$ and $1$, respectively. Introducing the vectors
$y = (0_{p}, 1_{N})$,
$\alpha = (\tilde{\alpha}, 0_{N})$
and 
$$
\log \eta = (
m(\widetilde{s}_1) + u(\widetilde{s}_1), \ldots, 
m(\widetilde{s}_p) + u(\widetilde{s}_p), 
m(s_1) + u(s_1), \ldots, m(s_N) + u(s_N)), 
$$
we can write the likelihood as 
\[
\pi(y \mid \lambda) \approx \exp(|\Gamma|) \prod_{i=1}^{N+p} \eta_i^{y_i} \exp(-\alpha_i \eta_i),
\]
which resembles the likelihood of observing \(N + p\) conditionally independent Poisson random variables with means \(\alpha_i \eta_i\) and observed values \(y_i\). Further,  $\log\eta$ is a multivariate Gaussian variable, which means that this is a latent Gaussian model. Thus, we can easily implement this approximation in Bayesian inference software such as R-INLA, provided that we can also compute the mean and precision matrix of $\log\eta$.

The mean value vector of $\log\eta$ is
$(
m(\widetilde{s}_1), \ldots, m(\widetilde{s}_p), m(s_1), \ldots, m(s_N))$. To evaluate the precision matrix, first note that by \citet{GRAPH1}, we can add vertices of degree 2 to $\Gamma$ without changing the distribution of the process and  let $\bar{\Gamma}$ denote the graph where all observation and integration locations have been added as vertices, by  
subdividing edges at these points to obtain a graph with a vertex set 
$$
\bar{\mathcal{V}} = \mathcal{V}\cup\{\widetilde{s}_1, \ldots, \widetilde{s}_p, s_1,\ldots,s_N\}.
$$
We can now use the methods from \citet{GRAPH4} to compute the precision matrix of the process at the vertices of $\bar{\Gamma}$ exactly if \(\alpha \in \mathbb{N}\). The reason for this is that the Whittle--Mat\'ern fields have Markov properties in this case \citep{GRAPH9}. For example, for $\alpha = 1$, the precision matrix is 
\[
Q_{ij} = 2\,\kappa\,\tau^{2}
\begin{cases}
\displaystyle\sum_{e \in \mathcal{E}_{v_i}}\!\Bigl(\tfrac{1}{2}+\tfrac{e^{-2\kappa l_e}}{1 - e^{-2\kappa l_e}}\Bigr) \mathbb{I}(\bar{e}\neq \underline{e}) 
\;+\;\tanh\Bigl(\frac{\kappa l_e}{2}\Bigr)\mathbb{I}(\bar{e}=\underline{e})
& \text{if } i=j,\\
\displaystyle\sum_{e \in \mathcal{E}_{v_i}\cap\mathcal{E}_{v_j}}-\frac{e^{-\kappa l_e}}{1 - e^{-2\kappa l_e}}
& \text{if } i \neq j.
\end{cases}
\]

Writing this precision matrix as a block matrix with blocks corresponding to the original vertices and the added locations, 
\[
\mathbf{Q} = \begin{pmatrix}
\mathbf{Q}_{\mathrm{vv}} & \mathbf{Q}_{\mathrm{vs}} \\
\mathbf{Q}_{\mathrm{sv}} & \mathbf{Q}_{\mathrm{ss}}
\end{pmatrix}.
\]
we can obtain the desired precision matrix as 
$\mathbf{Q}_{\mathrm{ss}} - \mathbf{Q}_{\mathrm{sv}}\mathbf{Q}_{\mathrm{vv}}^{-1}\mathbf{Q}_{\mathrm{vs}}.$
For higher $\alpha\in\mathbb{N}$, the methodology is similar but the derivations are more involved, see \citet{GRAPH4} for details.

The proposed method for approximating the posterior focuses on approximating only the integral in \eqref{eq:lgcp_likelihood}, while the latent field $u$ is evaluated exactly. This approach applies to cases where $\tau$ is either constant or given by $\tau(s) = \sigma^{-1} \sigma_\kappa(s)$. 
Let us now turn to the theoretical analysis of this approximation. Let $\mu$ be the posterior measure of the latent field $u$ given the data $Y$, and $\mu_p$ be the posterior measure induced by doing a numerical approximation of the integral of the latent field, where here $p$ represents the number of integration points. The accuracy of the approximate posterior measure will be measured with respect to the Hellinger distance, which is defined for two measures $\mu$ and $\mu_p$, absolutely continuous with respect to a common reference measure $\mu_0$, as
$$
d_{\text{Hell}}(\mu, \mu_p) = \left[ \frac{1}{2} \int \left\{ \left(\frac{d\mu}{d\mu_0}\right)^{1/2} - \left(\frac{d\mu_p}{d\mu_0}\right)^{1/2} \right\}^2 d\mu_0 \right]^{1/2}.
$$

The Hellinger distance plays a central role in quantifying the error between measures. Additionally, as established in \cite{stuart2010inverse} and noted in \cite{LGCP3}, convergence in Hellinger distance implies convergence in the total variation metric. Furthermore, this distance relates to the convergence of functionals through the inequality:
$$
\left| \mathbb{E}_\mu\{f(u)\} - \mathbb{E}_{\mu'}\{f(u')\} \right| \leq 2 \left[ \mathbb{E}_\mu\{|f(u)|^2\} - \mathbb{E}_{\mu'}\{|f(u')|^2\} \right] d_{\text{Hell}}(\mu, \mu').
$$

The following theorem formalizes the convergence result for our approximation.

\begin{theorem}\label{thm:hellinger_convergence}
    Let $\tau$ be either constant or given by $\tau(s) = \sigma^{-1} \sigma_\kappa(s)$, and let $u$ be the solution to \eqref{eq:WM_graph}. Let $\mu$ be the posterior measure of the latent field given the data $Y$, and let $\mu_p$ be the posterior measure obtained by approximating the integral in \eqref{eq:lgcp_likelihood} using $p$ integration points. 
    If $\alpha = 1$, $d_{\text{Hell}}(\mu, \mu_p) = \mathcal{O}(p^{-\gamma})$ for any $0 < \gamma < \frac{1}{2}$. 
    For $\alpha = 2$, $d_{\text{Hell}}(\mu, \mu_p) = \mathcal{O}(p^{-1}).$ 
\end{theorem}

The proof of this theorem, along with a more rigorous and precise statement, is deferred to Appendix \ref{app:proof_theorem}. Further, in Appendix \ref{app:proof_theorem} we prove convergence for a more general class of functions $\tau$ which may be of interest for models with non-stationary features that can be captured through a spatially varying $\tau$. For simplicity, the result is proved assuming $m=0$, as the most important object in the posterior measure is the latent field.

\section{Application to Road Accident Data}\label{sec:application}
\subsection{Introduction}
In this section, we illustrate the usage of the Whittle--Mat\'ern log-Gaussian Cox processes through an application to road accident data from Al-Ahsa, which is the largest governorate in Saudi Arabia's Eastern Province, with a population of over one million.
All analyses were conducted using R \citep{SOFTWARE1} and the MetricGraph, and R-INLA packages \citep{SOFTWARE3, SOFTWARE12} for modeling, while the packages ggplot2, sf, and excursions \citep{SOFTWARE5, SOFTWARE7, SOFTWARE8} supported visualization and exceedance probability calculations. 

\subsection{Data Description}
 The data was obtained from the Traffic Police Department in Dammam and was previously analysed by \citet{ACCIDENT1}. The dataset covers the period from October 2014 to August 2017 and initially included 3994 accident records. After data cleaning, 3563 records remained, consisting of eight key variables and covering the entire region along with highway connections to major cities. Accidents are classified into 46 causes, 19 types, and four severity levels, although only the two most severe categories are present in this subset. Precise geolocations, dates, and information on fatalities and serious injuries are also included, see \citet{ACCIDENT1}. In this study, we excluded variables related to accident type, severity, or cause, focusing solely on spatial modeling.

\begin{figure}[t]
    \centering
    \includegraphics[width=\textwidth]{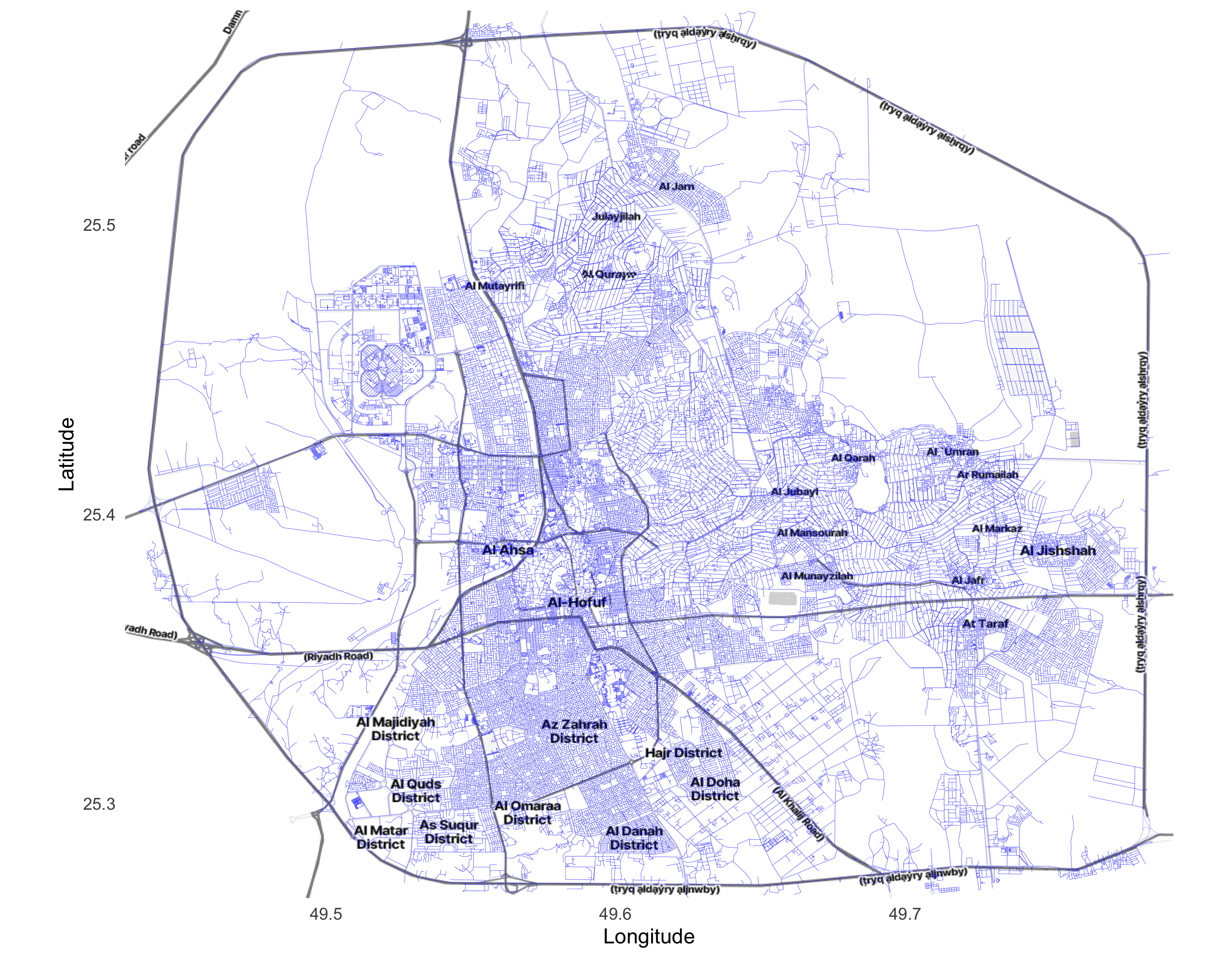}
    \caption{Street network representing the metric graph used for model fitting.}
    \label{fig:ahsa_graph}
\end{figure}

The road network was sourced from the TomTom Traffic Stats product \citep{ACCIDENT2}, providing additional data such as speed limits, road lengths, and functional road class. Functional road class ranks roads by importance, from 0 (major highways) to 7 (minor roads). For this study, we merged road classes 0 and 1 due to similar speed limits and the limited presence of Class 0 roads in the study region. To capture the full complexity of traffic dynamics and demonstrate the model’s scalability in handling extensive road networks, we included all road classes, although similar studies may focus on a subset of classes (1--6) for practicality.

Due to unclear village boundaries, we manually delineated the study area, balancing sufficient observations with manageable computational complexity. Therefore from the original dataset, we filtered 2482 accidents in the central part of Al-Ahsa over the three-year study period. The resulting area and network can be seen in Figure~\ref{fig:ahsa_graph}.

\subsection{Spatial Model Construction}

We fit two log-Gaussian Cox process models to analyze the spatial distribution of road accidents. The first, serving as a baseline, included only an intercept and a spatial random effect for the log intensity:
\begin{equation}
    \log(\lambda(s)) = \beta_0 + u(s), \label{eq:model1}
\end{equation}
where  \(\beta_0\) is the intercept, and \(u(s)\) is the spatial random effect modeled using a Whittle--Mat\'ern field with $\alpha=1$. The second model extended this formulation by incorporating spatial covariates for the log intensity:
\begin{equation}
\log(\lambda(s)) = \beta_0 + \sum_{i=1}^{12} \beta_i X_i(s) + u(s), \label{eq:model2}
\end{equation}
where \(X_i(s)\) represents the set of spatial covariates at location \(s\), and \(\beta_i\) are their associated coefficients.

Table~\ref{tab:model_covariates} describes all covariates used in the study. Data for these covariates were sourced from the osmdata R package \citep{SOFTWARE4}. The Euclidean distance from each mesh node to the nearest amenity was calculated. Although geodesic distances would likely improve accuracy, their computational demands were impractical given the large number of nodes in the graph. 

For the spatial random effect \(u(s)\), we used the MetricGraph package’s default priors for the Whittle--Mat\'ern field. Let $D_\Gamma$ be the length of the diagonal of the bounding box that contains $\Gamma$ as a planar object. The hyperparameters \(\kappa\) and \(\tau\) were both assigned log-Gaussian priors: \(\log(\kappa)\) was centred at \(\log(\kappa_0)\) with a precision of 10, where $\kappa_0$ is obtained based on the size of the spatial graph as $2/D_\Gamma$ and serves as starting value for $\kappa$ in the optimization, and \(\log(1/\tau)\) was centred at \(\log(2) - \log(\kappa_0)\), and assigned a precision of 10.  These defaults stabilize inference while allowing flexibility for the data to inform the spatial dependence structure.   

To ensure comparability between covariates, all variables were normalized using min-max scaling.  Categorical variables, such as road class, were encoded using one-hot encoding. The correlation matrices were analysed to detect multicollinearity; in cases of high correlation, variables were excluded to improve the robustness of the model. To emphasize the stronger influence of closer features, we applied a negative exponential transformation to distances to amenities, ensuring that closer amenities had a greater impact on the model. This transformation was applied consistently across all covariates. The traffic intensity was calculated following \citet{LGCP2} as: $\text{Traffic Intensity} = (\text{Speed Limit} / \text{Type of Road}) \times \log(\text{Road Length})$.

\begin{table}
    \centering
    \def~{\hphantom{0}} 
    \caption{Definition of spatial covariates}
    \label{tab:model_covariates}
    \begin{tabular}{lll}
        \hline
        \textbf{Covariate} & \textbf{Source} & \textbf{Description} \\ 
        \hline
        Road length & TomTom & Length of the road segment (in meters) \\
        Functional road class (FRC) & TomTom & Road classification (1: high to 7: low) \\
        Speed limit & TomTom & Speed limit (in km/h) \\
        Traffic intensity & TomTom & Measure of traffic density \\
        Hospitals & OpenStreetMap & Number of healthcare facilities \\
        Education & OpenStreetMap & Number of educational institutions \\
        Finance & OpenStreetMap & Number of financial institutions \\
        Mosques & OpenStreetMap & Number of mosques \\
        Intersections & OpenStreetMap & Number of intersections \\
        \hline
    \end{tabular}
\end{table}

The graph, constructed using road segments from TomTom Traffic Stats, resulted in a metric graph with 81,438 vertices and 178,801 edges, providing comprehensive coverage of the study area’s road network (see Figure~\ref{fig:ahsa_graph}). Given the data’s high precision, we optimized the graph-building process in the MetricGraph package by setting \texttt{merge\_close\_vertices = FALSE}. This allowed the graph to be constructed in under 1.5 minutes on a machine with an Apple M1 Max chip and 32 GB of RAM. Further optimization was achieved by pruning vertices of degree 2, leveraging the fact that vertices of degree 2 can be removed without changing the model \citep{GRAPH1}. 
This step reduced the number of degree 2 vertices from 20,502 to 6,960. The final graph used for model fitting included 67,896 vertices and 165,259 edges. 

For the numerical approximation of the integral in the likelihood \eqref{eq:lgcp_likelihood}, a mesh with 100-meter spacing was used, resulting in 146,237 mesh nodes to ensure adequate spatial resolution.

\subsection{Results}

\begin{table}
    \centering
    \caption{Gaussian Process Parameters for log-Gaussian Cox process Models}
    \label{tab:gp_parameters}
    \def~{\hphantom{0}} 
    \begin{tabular}{lcccc}
        \hline
        \textbf{Parameter} & \textbf{Model 1 Mean} & \textbf{Model 1 95\% CRI} & \textbf{Model 2 Mean} & \textbf{Model 2 95\% CRI} \\ 
        \hline
        $\kappa$ & 0.209  & [0.133, 0.302] &  0.189  & [0.117, 0.276] \\
        $\tau$ & 0.213   & [0.202, 0.226] & 0.216 & [0.204, 0.209] \\
        Intercept ($\beta_0$) &  -2.630 & [-3.032, -2.243] & -6.214 & [-8.333, -4.095]  \\
        \hline
    \end{tabular}
    \vspace{0.5em}
        \begin{minipage}{0.9\textwidth}  
    \centering
        \footnotesize{CRI: credible interval. Parameters include $\kappa$, $\tau$ for the Whittle-Matérn field, and intercept $\beta_0$.}
    \end{minipage}
  
\end{table}
Both models were estimated using R-INLA. This required 55 seconds for Model 1 and 73 seconds for Model~2. 
The estimated field parameters and the estimated intercepts of the two models are shown in Table~\ref{tab:gp_parameters}.
Table~\ref{tab:model_results} summarizes the parameter estimates from Model 2. For each covariate, we present the estimated coefficients and their 95\% credible intervals. For example, we note that the coefficient for traffic intensity is positive, which is reasonable, as it indicates a higher risk of accidents on busy roads.

\begin{table}
    \centering
    \caption{Coefficient estimates for log-Gaussian Cox process Model 2}
    \label{tab:model_results}
    \def~{\hphantom{0}} 
    \begin{tabular}{lccc}
        \hline
        \textbf{Covariate} & \textbf{Mean} & \textbf{Lower 95\% CRI} & \textbf{Upper 95\% CRI} \\ 
        \hline
        Traffic intensity         & 8.278  & 3.594  & 12.963 \\
        FRC1                      & -0.688 & -1.067 & -0.310 \\
        FRC2                      & -8.275 & -28.637 & 12.086 \\
        FRC3                      & -1.124 & -1.542 & -0.706 \\
        FRC4                      & -2.893 & -3.883 & -1.904 \\
        FRC5                      & -0.739 & -1.017 & -0.461 \\
        FRC6                      & -0.252 & -0.421 & -0.083 \\
        Distance to mosques       & -0.008 & -0.245 & 0.229 \\
        Distance to education     & -0.667 & -0.917 & -0.417 \\
        Distance to finance       & 1.084  & 0.854  & 1.314 \\
        Distance to hospitals     & -0.717 & -1.010 & -0.424 \\
        Distance to intersections & 0.115  & -0.119 & 0.348 \\
        \hline
    \end{tabular}
    \vspace{0.5em}
    \begin{minipage}{0.9\textwidth}  
    \centering
        \footnotesize CRI: credible interval; FRC: functional road class (one-hot encoded).
    \end{minipage}
\end{table}

The posterior mean of the log intensities for the two models is shown in Figure~\ref{fig:model_results}, highlighting the high-risk regions identified by the models. We can note that these two surfaces are quite similar. Figure~\ref{fig:model2_u} shows the posterior mean of the Gaussian field $u$ for Model 2, indicating the spatial variation not captured by the covariates. We can note that the covariates fail to capture the higher risk of accidents in the central part of the region and along some of the main roads.

\begin{figure}[t]
    \centering
    \includegraphics[width=\textwidth]{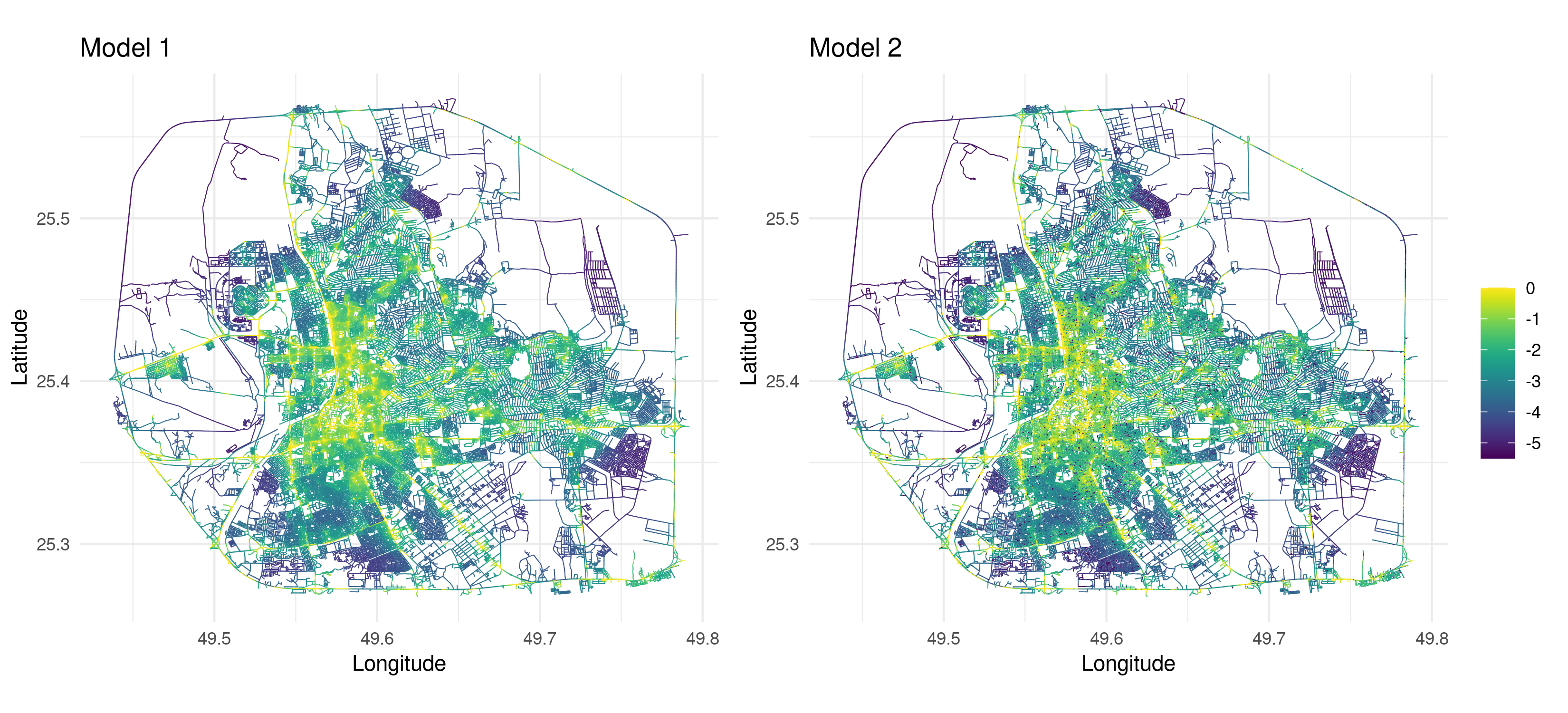}
    \caption{The estimated log-intensity $\log \lambda(s)$ for Model 1 (left) and Model 2 (right).}
    \label{fig:model_results}
\end{figure}

\begin{figure}[t]
    \centering
    \includegraphics[width=\textwidth]{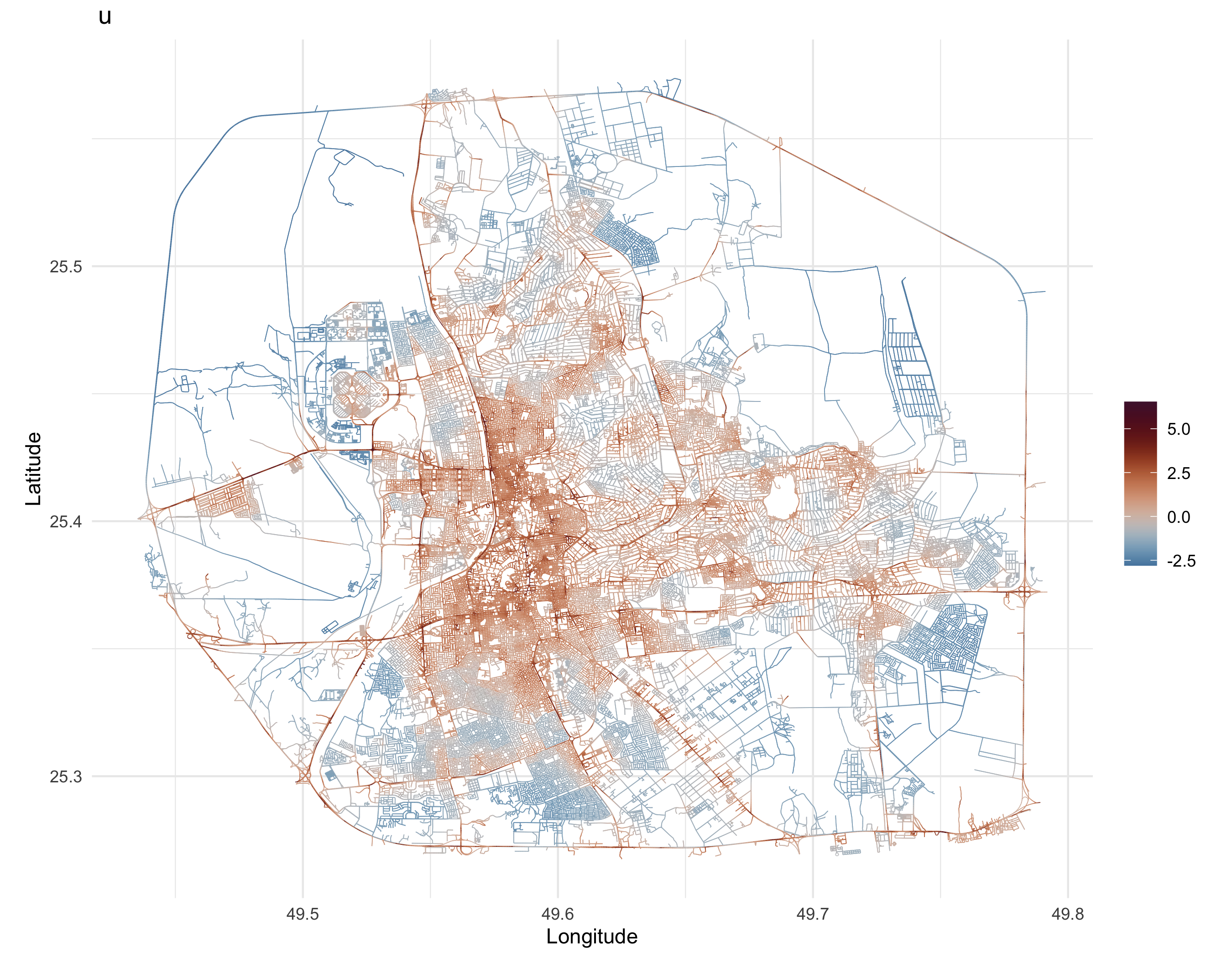}
    \caption{The posterior mean of $u$ for Model 2.}
    \label{fig:model2_u}
\end{figure}

\section{Excursion Set Analysis and Risk Localization}\label{sec:excursions}

Conventional applications of log-Gaussian Cox process models often conclude with the posterior mean $\mathbb{E}(u \mid Y)$ and variance, providing point estimates of the latent field. However, for many practical applications, the purpose is to find high risk areas, or so-called hotspots. These can be estimated by calculating excursion sets of the latent process, which we introduce for metric graph models in this section and then apply to the traffic accident analysis. For traffic applications, identifying these zones is critical for informing policymakers and practitioners as they develop targeted interventions to reduce road accidents and improve traffic safety.

A method for defining simultaneous credible bands for latent Gaussian models was introduced by \citet{RISK1} and implemented in the excursions R package \citep{SOFTWARE8}. Given that log-Gaussian Cox process models fall within the latent Gaussian model family, we can apply this method to identify high-risk areas within the road network. The integration of this approach with INLA model objects ensures computational efficiency, a significant advantage in large-scale applications like ours.

Our goal is to identify a region $D$ on the graph $\Gamma$ where the process $u(s)$ in \eqref{eq:WM_graph} exceeds a threshold $t$ with a probability of $1 - \alpha$ for all $s \in D \subset \Gamma$. This so-called positive excursion set indicates high-risk areas not captured by the covariates, where additional safety measures might be required.
If $u(s) = f(s)$ was a known function, we could compute the positive excursion set directly as $A_t^{+}(f) = \{ s \in \Gamma: f(s) > t \}$. However, because $u(s)$ is a latent random process, we focus on identifying regions where the process exceeds the threshold with high probability. Specifically, the positive excursion set with probability $\alpha$, $E_{t, \alpha}^{+}(u)$, is defined as
\begin{equation} \label{eq:excursion_set}
E_{t, \alpha}^{+} = \underset{D}{\arg \max} \left\{ |D| : \mathbb{P} \left[ D \subset A_t^{+}(u) \right] \geq 1 - \alpha \right\},
\end{equation}
The interpretation of equation~\eqref{eq:excursion_set} is that it represents the largest set where, with probability $1 - \alpha$, the process exceeds the threshold level $t$ at all locations within the set. 

In practice, the excursion sets can be computed efficiently for several values of $\alpha$ by parametrizing the possible excursion regions through excursion sets of the marginal excursion probabilities and then computing the excursion function 
$F_t(s) = \sup \{ 1 - \alpha : s \in E_{t, \alpha}^{+} \}$,
see \citet{RISK1} for details. The excursion function takes values in $[0,1]$ and an excursion set for a specific value of $\alpha$ is obtained as $A_{1-\alpha}^{+}(F_t)$.

To compute an excursion set, we must decide on the level $t$ to use. For Model 2 in the application, we set $t=0$ and compute the excursion function of the latent process $u$. Thus, this identifies areas where the Gaussian process has a significant contribution to increasing the accident risk. 

The result is shown in Figure~\ref{fig:excursion_results2}, where one can note that a few roads are highlighted as being unexpectedly dangerous taking the covariates into account.  Figure~\ref{fig:excursion_results} shows the marginal excursion probabilities computed for Model 2.

\begin{figure}[t]
    \centering
        \includegraphics[width=\textwidth]{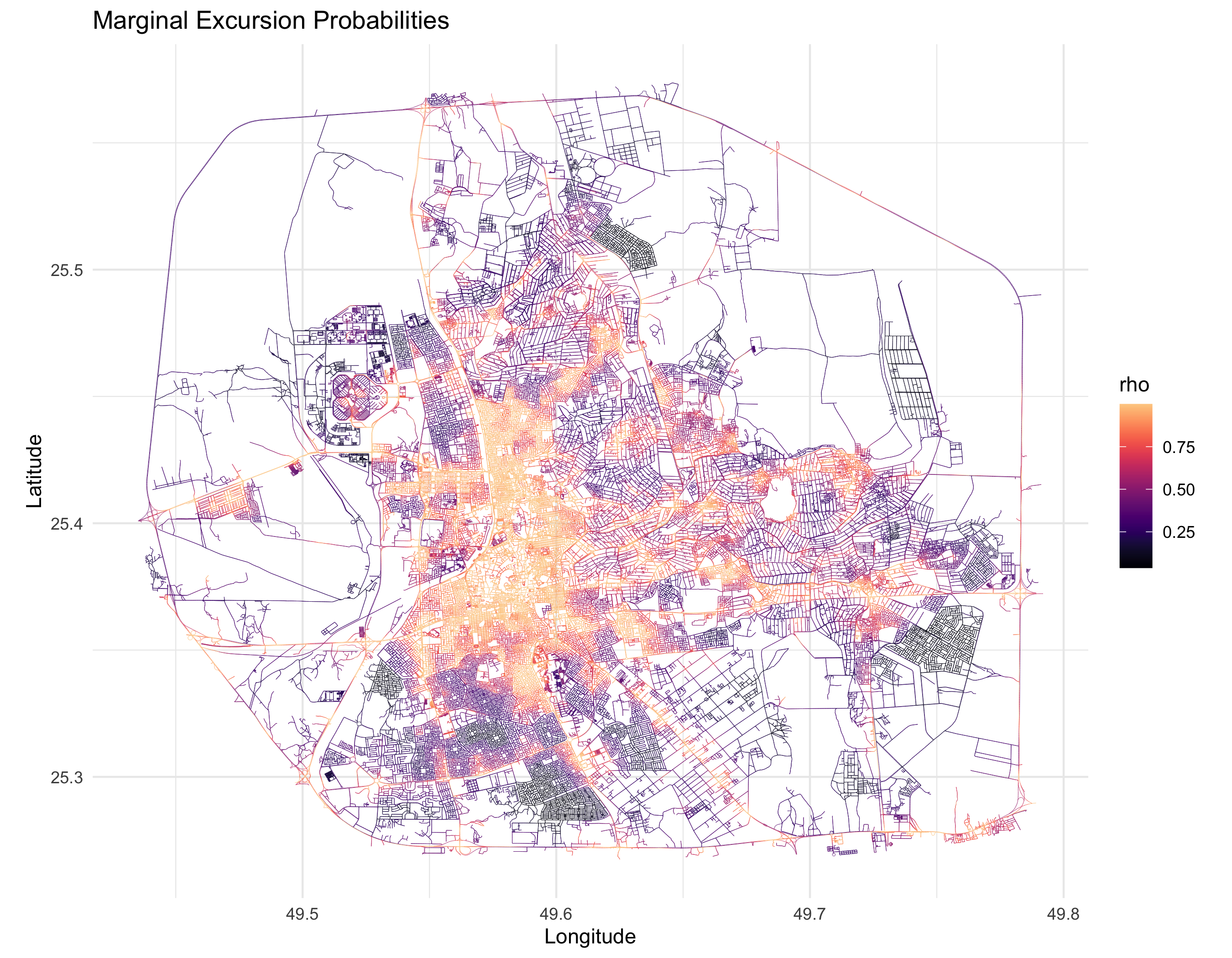}
    \caption{Marginal excursion probabilities for $u$ in Model 2.}
    \label{fig:excursion_results}
\end{figure}

\begin{figure}[t]
    \centering
        \includegraphics[width=0.9\textwidth]{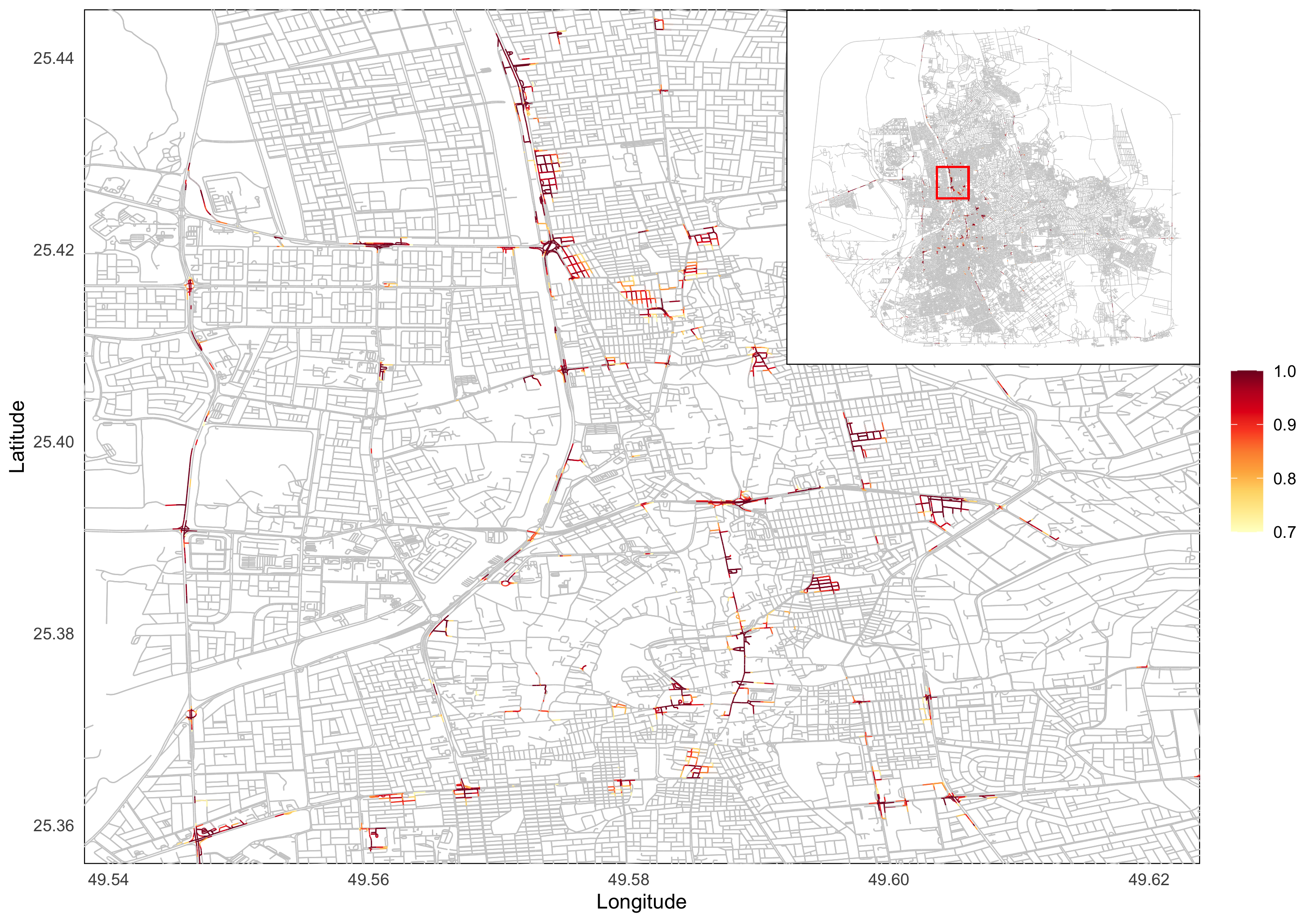}
    \caption{Excursion function for $u$ in Model 2.}
    \label{fig:excursion_results2}
\end{figure}

\section{Discussion and Future Work}\label{sec:discussion}

We have introduced the first class of log-Gaussian Cox processes which are well-defined for any metric graph and showed how to perform computationally efficient likelihood-based inference for these models. 
The models are based on the recently introduced Gaussian Whittle--Mat\'ern fields, and the key feature we used for efficient inference is that finite dimensional distributions of these fields can be evaluated exactly and correspond to Gaussian graphical models \citep{GRAPH13}. 
Furthermore, we have provided rigorous results regarding the rate of convergence of the approximated posterior distributions to the true posterior distribution with respect to the Hellinger distance. This is, to the best of our knowledge, the first complete proof of such convergence on log-Gaussian Cox processes with the Gaussian process being driven by SPDEs.

To showcase the practical utility of the model, we applied it to traffic accident data on a large scale network, performing hotspot localization to identify high-risk road segments by combining the log-Gaussian Cox process with exceedance probabilities. This methodology allows for precise identification of high-risk areas, as illustrated in Figure~\ref{fig:excursion_results2}. 

The application also demonstrates the feasibility of implementing these models on extensive graphs that encapsulate the dynamics of an entire city or region within a single framework. 
While it is technically feasible to apply the model proposed in \citet{LGCP5} to this dataset, doing so would require linearizing the graph structure shown in Figure~\ref{fig:ahsa_graph}, which introduces computational challenges. For instance, the metric graph in our model consists of 165,259 edges, whereas a linearized version would require 532,122 lines to maintain equivalent network geometry. Further, as isotropic models cannot be Markov on general graphs such as those \citep{GRAPH9}, the resulting log-Gaussian Cox process would not have any sparsity to reduce the computational costs, which would make it computationally infeasible to fit on any standard desktop computer. This comparison underscores the computational efficiency of our approach.

Finally, while excursion sets have been extensively applied in environmental studies, such as in precipitation or temperature predictions in spatial contexts \citep{RISK2}, their application to log-Gaussian Cox process models remains relatively unexplored. To the best of our knowledge, this is the first use of excursion sets in a log-Gaussian Cox process context for road accident analysis. 

A natural question for future work is how to extend the models to accommodate general values of $\alpha$. A natural idea here is to use the methods of \citet{GRAPH11}, where a finite element approximation of the process is used in combination with a rational approximation of the fractional power. Through the methods in \citet{GRAPH11} we could also allow the parameters $\kappa$ and $\tau$ to be spatially varying functions to model nonstationary features. 
This would, however, require extending our theoretical results to account for the additional approximations of the Gaussian field. 
Another direction for future work is investigating the appropriateness of the Kirchhoff conditions at vertices, as defined in equation~\eqref{eq:Kirchoff}. These may not align perfectly with traffic dynamics, since the influx from smaller roads into larger networks results in a nonzero total flow. A possible solution is to adopt a ``weighted Kirchhoff" approach, assigning different weights to roads based on their importance. 
Future research will also explore space--time extensions of the model and multivariate extensions for multi-type point patterns.



\section*{Acknowledgement}
We appreciate the data shared by the authors of \citet{ACCIDENT1}.

\appendix
\section{A review on approximation of posterior measures}\label{app:post_meas}

In this appendix, we will summarise some of the results found in \cite{cotter2010approximation}. We begin by outlining the basic framework. Let $\mu_0$ be a centred Gaussian measure, which we will refer to as the prior measure. Let $(X,\|\cdot\|_X)$ be a Banach space such that $\mu_0(X) = 1$. Further, assume that we are interested in a posterior measure $\mu$ that is absolutely continuous with respect to $\mu_0$ and has Radon--Nikodym derivative given by
$$\frac{d\mu}{d\mu_0}(u) = \frac{1}{Z(y)} \exp\{-\Phi(u;y)\},$$
where $y\in Y$ is some observed data, $(Y,\|\cdot\|_Y)$ is a Banach space, $\Phi(\cdot;\cdot)$ is called the potential, and is determined by the distribution of $y$ given $u$, and the normalization constant $Z(y)$ is given by
\begin{equation}\label{eq:norm_const}
    Z(y) = \int_X \exp(-\Phi(u, y)) \, d\mu_0(u).
\end{equation}

Now, we state the assumptions on the potential:

\begin{assumption}\label{assump:assumptions_Cotter_Stuart}
    For some Banach space $X$ with $\mu_0(X) = 1$, the function $\Phi: X \times Y \to \mathbb{R}$ satisfies the following:
    \begin{enumerate}
        \item[(i)] For every $\varepsilon > 0$ and $r > 0$, there exists $M = M(\varepsilon, r) \in \mathbb{R}$ such that for all $u \in X$ and $y \in Y$ with $\|y\|_Y < r$, 
        \[
        \Phi(u; y) \geq M - \varepsilon \|u\|_X^2.
        \]
        \item[(ii)] For every $r > 0$, there exists $L = L(r) > 0$ such that for all $u \in X$ and $y \in Y$ with $\max\{\|u\|_X, \|y\|_Y\} < r$, 
        \[
        \Phi(u; y) \leq L(r).
        \]
    \end{enumerate}
\end{assumption}

To approximate the measure $\mu$, we approximate the potential $\Phi$ within a $N$-dimensional subspace of $X$. Specifically, the approximate measure $\mu^N$ is defined as:
$$
\frac{d\mu^N}{d\mu_0}(u, y) = \frac{1}{Z^N(y)} \exp(-\Phi^N(u, y)),
$$
where the normalization constant $Z^N(y)$ is given by:
\begin{equation}\label{eq:norm_const_appr}
Z^N(y) = \int_X \exp(-\Phi^N(u, y)) \, d\mu_0(u).
\end{equation}

The potential $\Phi^N$ provides a simplified representation of the forward problem restricted to the $N$-dimensional subspace of $X$. The relationship between the approximation of $\Phi$ and the measures $\mu$ and $\mu^N$ has been rigorously studied in \cite{cotter2010approximation}, where it is shown that under appropriate assumptions, errors in $\Phi^N$ lead to controlled discrepancies between $\mu$ and $\mu^N$.

To measure the closeness of $\mu$ and $\mu^N$, we use the Hellinger distance. Under these tools, the following theorem is proved in \cite[Theorem 2.2]{cotter2010approximation}.
\begin{theorem}\label{thm:cotter_stuart}
    Assume that $\Phi$ and $\Phi^N$ satisfy Assumption \ref{assump:assumptions_Cotter_Stuart} with constants uniform in $N$. Further, assume that for any $\varepsilon > 0$, there exists $K = K(\varepsilon) > 0$ such that
    $$|\Phi(u;y) - \Phi^N(u;y)| \leq K \exp\left(\varepsilon \|u\|_X^2\right) \psi(N),$$
    where $\psi(N) \to 0$ as $N \to \infty$. Then, the measures $\mu$ and $\mu^N$ are close with respect to the Hellinger distance. Specifically, there exists a constant $C$, independent of $N$, such that
    $$d_{\text{Hell}}(\mu, \mu^N) \leq C \psi(N).$$
\end{theorem}

This result implies that all moments of $\|u\|_X$ are $\mathcal{O}(\psi(N))$ close. In particular, when $X$ is a Hilbert space, the mean and covariance operator are also $\mathcal{O}(\psi(N))$ close.

\section{Auxiliary results}\label{app:aux_res}
In this appendix, we will prove auxiliary results needed for the application of Theorem \ref{thm:cotter_stuart}. More precisely, we will describe the Gaussian measure we will work with, along with the auxiliary Banach spaces and properties of the Gaussian measures.

First, let $\Gamma$ be a compact and connected metric graph, and recall the definition of $L_2(\Gamma)$ provided in the main text. We will also need to introduce some additional spaces. First, recall that $C(\Gamma)$ is the set of real-valued continuous functions in $\Gamma$, which is endowed with the norm $\|f\|_{C(\Gamma)} = \sup_{s\in\Gamma} |f(s)|$. Further, we need the space of $\gamma$-Hölder continuous functions, $ C^{0,\gamma}(\Gamma) $, for $ 0 < \gamma \leq 1 $, with seminorm 
$[f]_{C^{0,\gamma}(\Gamma)} = \sup_{s, s' \in \Gamma} \frac{|f(s) - f(s')|}{d(s, s')^\gamma},$ and norm $\|f\|_{C^{0,\gamma}(\Gamma)} = \|f\|_{C(\Gamma)} + [f]_{C^{0,\gamma}(\Gamma)}.$ We also need the following Banach space,
$$\widetilde{C}^1(\Gamma) = C(\Gamma) \cap \left(\bigoplus_{e\in\mathcal{E}} C^1(e)\right),$$
with its natural norm
$$\|f\|_{\widetilde{C}^1(\Gamma)} = \|f\|_{C(\Gamma)} + \sum_{e\in\mathcal{E}} \| f_e'\|_{C(e)}.$$

Now, we will show some numerical integration bounds for functions in these spaces. 

\begin{lemma}\label{lem:holder_num_int_bound}
Let $0<\gamma\leq 1$, and let $f\in C^{0,\gamma}(e)$, with $e=[0,l_e]$, $l_e>0$. Further, let $0\leq a \leq t^* \leq b \leq l_e$. Then,
$$\int_a^b |f(t) - f(t^*)| dt \leq \frac{2}{1+\gamma}\|f\|_{C^{0,\gamma}(e)} |b-a|^{1 + \gamma}.$$
\end{lemma}
\begin{proof}
We have that
\begin{align*}
    \int_a^b |f(t) - f(t^*)| dt &\leq \|f\|_{C^{0,\gamma}(e)} \int_a^b |t - t^*|^\gamma dt = \|f\|_{C^{0,\gamma}(e)}\left( \int_a^{t^*} (t^* - t)^\gamma dt + \int_{t^*}^b (t - t^*)^\gamma dt\right)\\
    &=  \|f\|_{C^{0,\gamma}(e)} \left(\frac{(b - t^*)^{1 + \gamma}}{1 + \gamma} + \frac{(t^* - a)^{1 + \gamma}}{1 + \gamma}\right)   \leq \frac{2}{1+\gamma} \|f\|_{C^{0,\gamma}(e)} |b-a|^{1+\gamma}.
\end{align*}
This concludes the proof.
\end{proof}

By proving in a similar manner, but using the mean value theorem, we have that
\begin{lemma}\label{lem:C1_num_int_bound}
    Let $f\in C^1(e)$, with $e=[0,l_e]$, $l_e>0$. Further, let $0\leq a \leq t^* \leq b \leq l_e$. Then,
    $$\int_a^b |f(t) - f(t^*)| dt \leq \|f\|_{C^{1}(e)} |b-a|^{2}.$$
\end{lemma}

Now, let $u$ be the solution to \eqref{eq:WM_graph}. In order to be able to accommodate the variance-stationary model, we will allow $\tau$ to vary in space, and we will consider the following set of assumptions on $\tau$:

\begin{assumption}\label{assump:tau_assump}
    Let $1/2 < \alpha \leq 2$ and define $\tilde{\alpha} = \min\{\alpha - 1/2, 1/2\}$.     
    \begin{enumerate}
        \item[(i)] If $\alpha > 1/2$, then for some $\gamma \in (0, \tilde{\alpha})$, we have $\tau \in C^{0,\gamma}(\Gamma)$.
        
        \item[(ii)] If $\alpha > 3/2$, then for some $\gamma \in (0, \tilde{\alpha} - 1)$, $\tau_e \in C^{1,\gamma}(e)$ for all $e \in \mathcal{E}$, and $\tau$ satisfies the Kirchhoff conditions, that is, $\tau \in C(\Gamma)\cap K(\Gamma)$.
    \end{enumerate}
\end{assumption}

We now summarize some of the results related to solutions to \eqref{eq:WM_graph} with varying $\tau$. The proof of the following result can be found in \cite[Propositions 2 and 3]{GRAPH11}. 

\begin{proposition}
    \label{prop:regularity_results}
    Let Assumption \ref{assump:tau_assump} hold, and let $u$ be the solution to \eqref{eq:WM_graph}. Then, $u$ admits a modification $\tilde{u}$ such that, with probability 1:
    \begin{enumerate}
        \item[(i)] If $\alpha > 1/2$, $\tilde{u} \in C^{0,\gamma}(\Gamma)$ for the value of $\gamma$ specified in Assumption \ref{assump:tau_assump}.
        \item[(ii)] If $\alpha > 3/2$, $\tilde{u}_e \in C^{1,\gamma}(e)$ for every $e \in \mathcal{E}$, and $u \in C(\Gamma) \cap K(\Gamma)$, where $\gamma$ is as given in Assumption \ref{assump:tau_assump}.
    \end{enumerate}
\end{proposition}

The next result, proved in \cite[Proposition 4]{GRAPH11}, shows that by choosing $\tau(s) = \sigma^{-1}\sigma_\kappa(s)$ to create the variance-stationary Whittle--Mat\'ern field, the resulting $\tau$ satisfies Assumption \ref{assump:tau_assump}, and thus, the solutions have the regularity given in Proposition \ref{prop:regularity_results}.

\begin{proposition}
    \label{prop:tau_and_u}
    Let $1/2 < \alpha \leq 2$. Assume $\tau(s) = \sigma^{-1} \sigma_\kappa(s)$, where $\sigma_\kappa$ represents the marginal standard deviations of a Whittle--Mat\'ern field with $\tau = 1$ and $\sigma>0$ is constant. If $u$ is the solution to \eqref{eq:WM_graph}, then $u$ admits a modification $\tilde{u}$ such that:
    \begin{enumerate}
        \item[(i)] If $\alpha > 1/2$, then $\tau \in C^{0,\tilde{\alpha}}(\Gamma)$, where $\tilde{\alpha} = \min\{\alpha - 1/2, 1/2\}$, and $\tilde{u} \in C^{0,\gamma}(\Gamma)$ for every $0 < \gamma < \tilde{\alpha}$.
        \item[(ii)] If $\alpha > 3/2$, then $\tau_e \in C^{1,\alpha - 3/2}(e)$ for every $e \in \mathcal{E}$, and $\tau\in C(\Gamma)\cap K(\Gamma)$. Consequently, $u_e \in C^{1,\gamma}(e)$ for every $e \in \mathcal{E}$ and $0 < \gamma < \alpha - 3/2$, and $\tilde{u} \in C(\Gamma) \cap K(\Gamma)$.
    
    \end{enumerate}
\end{proposition}

From the above results, we can conclude that under Assumption \ref{assump:tau_assump}, the solution $u$ to \eqref{eq:WM_graph} admits a modification $\tilde{u}$ with the following properties. For $\alpha > 1/2$ and $0 < \gamma < \min\{\alpha - 1/2, 1/2\}$, we have 
$
\mathbb{P}(\tilde{u} \in C^{0,\gamma}(\Gamma)) = 1.
$
Furthermore, for \(\alpha > 3/2\), there exists a modification $\tilde{u}$ such that
$
\mathbb{P}(\tilde{u} \in \tilde{C}^1(\Gamma)) = 1.
$

It is worth noting that Assumption \ref{assump:tau_assump} is satisfied by both constant $\tau$ and $\tau$ derived from the variance-stationary model, ensuring the validity of these results for such choices of $\tau$.

We conclude this section by introducing the Gaussian measure we will be working with. More precisely, the prior Gaussian measure we will use, namely $\mu_0$, is the distribution of the modification $\tilde{u}$ of the solution $u$ to \eqref{eq:WM_graph} obtained in Propositions \ref{prop:regularity_results} and \ref{prop:tau_and_u}. Therefore, by if $\tau$ satisfies Assumption \ref{assump:tau_assump}, then
$$
\mu_0(C^{0,\gamma}(\Gamma)) = 1 \quad \text{for } \alpha > 1/2,
$$
where $\gamma$ is given in Assumption \ref{assump:tau_assump}, and
$$
\mu_0(\widetilde{C}^1(\Gamma)) = 1 \quad \text{for } \alpha > 3/2.
$$

\section{Proof of the main theorem}\label{app:proof_theorem}
We are now in a position to prove the convergence result. Throughout this section we will assume Assumption~\ref{assump:tau_assump} holds. Further, we will assume that the mean function $m(\cdot)$ is constant equal to zero. We start by recalling the definition of the Gaussian measure $\mu_0$ from Appendix \ref{app:aux_res}, as the distribution of $\tilde{u}$, which is the modification of the solution $u$ to \eqref{eq:WM_graph} such that for $\alpha > 1/2$, we have
$\mu_0(C^{0,\gamma}(\Gamma)) = 1,$
where $0 < \gamma < \min\{\alpha-1/2,1/2\}$ is given in Assumption \ref{assump:tau_assump}, and for $\alpha>3/2$, we have
$\mu_0(\widetilde{C}^1(\Gamma)) = 1.$
For simplicity, in the remaining of this section we will denote the modification $\tilde{u}$ by $u$. 

Next, we need to define the Banach space where our data will belong. To this end, we define the following Banach space:
$$Y = \mathbb{R}^{|\mathcal{E}|} \times \ell_\infty(\mathbb{N})^{|\mathcal{E}|},$$
where $\ell_\infty(\mathbb{N})$ is the Banach space of bounded sequences with the norm $\|x\|_{\ell_\infty(\mathbb{N})} = \sup_{i\in\mathbb{N}} |x_i|$, for a sequence $x = (x_i)_{i=1}^\infty$. 

Given a point pattern $P$, it is identified with the element $y = \left((n_e)_{e\in\mathcal{E}}, ((s_{i,e})_{i=1}^\infty)_{e\in\mathcal{E}}\right)$ of $Y$, where $n_e$ represents the number of points in edge $e$, and $s_{i,e}$ is the location of the $i$th point of $P$ in the edge $e$, and for $i>n_e$ we have $s_{i,e} = 0$. Further, $Y$ is endowed with its natural norm:
$$\|y\|_Y = \sum_{e\in\mathcal{E}} (|n_e| + \|s_e\|_{\ell_\infty(\mathbb{N})}),$$
where $s_e = (s_{i,e})_{i=1}^\infty$. 

For our log-Gaussian Cox process, we need the following potential:
$$\Phi(u;y) = \sum_{e\in\mathcal{E}} \left(\int_e \exp\{u_e(t)\}dt - \sum_{i=1}^{n_e} u(s_{i,e}) \right).$$
Then, in the same spirit as in \cite{LGCP3}, we have that the posterior measure $\mu$ for $u(\cdot)$ conditioned on $y$, can be defined as its Radon-Nikodym derivative 
\begin{equation}\label{eq:post_mu_expr}
    \frac{d\mu}{d\mu_0}(u) = Z(y)^{-1} \exp\{-\Phi(u;y)\}, 
\end{equation}
where $Z(y)$ is the normalizing constant given in \eqref{eq:norm_const}.

As discussed in Appendix \ref{app:post_meas}, approximating the posterior distribution reduces to approximating the potential $\Phi(\cdot; \cdot)$. We consider the following approximation:
$$
\Phi_p(u;y) = \sum_{e\in\mathcal{E}}\left(\sum_{i=1}^{p_e} w_{i,e} \exp\{u(\tilde{s}_{i,e})\} - \sum_{i=1}^{n_e} u(s_{i,e})\right),
$$
where $\tilde{s}_{i,e}$ are integration points and $w_{i,e}$ are weights, with $i=1,\ldots, p_e$, $p_e \in \mathbb{N}$, and $e\in\mathcal{E}$. The corresponding approximated posterior measure is given by its Radon-Nikodym derivative:
\begin{equation}\label{eq:post_mu_apprx_expr}
    \frac{d\mu_p}{d\mu_0}(u) = \frac{1}{Z_p(y)} \exp\{-\Phi_p(u;y)\},
\end{equation}
where $Z_p(y)$ is the normalizing constant defined in \eqref{eq:norm_const_appr}. 


The main result of this paper is stated in the following theorem.

\begin{theorem}\label{thm:main_theorem}
Let $u$ be the solution to \eqref{eq:WM_graph}, where $\tau$ satisfies Assumption \ref{assump:tau_assump} and $\alpha = 1$ or $\alpha = 2$. For each edge $e$, consider a partition $q_{1,e}, \ldots, q_{p_e+1,e}$, where $p_e \in \mathbb{N}$ and there exists a constant $K$, independent of $p_e$, such that there exist $K_1,K_2>0$ such that
$$
K_1 p_e^{-1} \leq  |q_{i+1,e} - q_{i,e}| \leq K_2 p_e^{-1},
$$
for $i=1,\ldots, p_e$, with $\tilde{s}_{0,e} = 0$. Define $\|p\| = \min_{e\in\mathcal{E}} p_e$, and let $\mu$ and $\mu_p$ be the posterior measures given by \eqref{eq:post_mu_expr} and \eqref{eq:post_mu_apprx_expr}, respectively. Assume that $\tilde{s}_{i,e} \in [q_{i,e}, q_{i+1,e}]$ for $i = 1,\ldots, p_e$ and $e\in\mathcal{E}$. Further, let $w_i = | \tilde{s}_{i+1,e} - \tilde{s}_{i,e} |$. Then, there exists a constant $C$, independent of $p$, such that
$$
d_{\text{Hell}}(\mu, \mu_p) \leq C \psi(\|p\|),
$$
where 
$$
\psi(\|p\|) =
\begin{cases}
    \|p\|^{-\gamma}, & \text{if } \alpha = 1, \text{ where } \gamma \text{ is given in Assumption \ref{assump:tau_assump}}, \\
    \|p\|^{-1}, & \text{if } \alpha = 2.
\end{cases}
$$
\end{theorem}

\begin{proof}
Our goal is to apply Theorem \ref{thm:cotter_stuart}. For this proof, for the case $\alpha = 1$, $\gamma$ will be given from Assumption~\ref{assump:tau_assump}. Observe that for $\alpha = 1$ we will work with the Banach space $X = C^{0,\gamma}(\Gamma)$, and for $\alpha = 2$ we will work with the Banach space $X = \widetilde{C}^1(\Gamma)$. Thus, $\mu_0(X) = 1$.

Next, we verify that both potentials $\Phi(\cdot;\cdot)$ and $\Phi_p(\cdot;\cdot)$ satisfy Assumption \ref{assump:assumptions_Cotter_Stuart}. We begin by proving condition (i) of Assumption \ref{assump:assumptions_Cotter_Stuart} for both potentials.

Let $r > 0$ and take any $y$ such that $\|y\|_Y < r$. This implies that $\sum_{e\in\mathcal{E}} n_e < r$. Using this, we obtain
\begin{align*}
    \Phi(u;y) &= \sum_{e\in\mathcal{E}} \left(\int_e \exp\{u_e(t)\}dt - \sum_{i=1}^{n_e} u(s_{i,e}) \right) \geq - \sum_{e\in\mathcal{E}} \sum_{i=1}^{n_e} u(s_{i,e}) \geq - \|u\|_{C(\Gamma)} \sum_{e\in\mathcal{E}} n_e\\
    &\geq - r \|u\|_{C(\Gamma)} \geq -\frac{r^2}{4\varepsilon} - \varepsilon \|u\|_{C(\Gamma)}^2,
    \end{align*}
where in the last inequality we applied Young's inequality with $\varepsilon > 0$, i.e., $ab \leq \varepsilon a^2 + b^2 / (4\varepsilon)$.

Now, note that for any $0 < \gamma \leq 1$, we have the bounds
$$
\|u\|_{C(\Gamma)} \leq \|u\|_{C^{0,\gamma}(\Gamma)}
\quad \text{and} \quad
\|u\|_{C(\Gamma)} \leq \|u\|_{\widetilde{C}^1(\Gamma)}.
$$
Let $M = M(\varepsilon,r) = -\frac{r^2}{4\varepsilon}$. Thus, for $\alpha = 1$, it follows that for every $\varepsilon > 0$,
$$
\Phi(u;y) \geq M - \varepsilon \|u\|_{C(\Gamma)}^2 \geq M - \varepsilon \|u\|_{C^{0,\gamma}(\Gamma)}^2,
$$
while for $\alpha = 2$, we obtain
$$
\Phi(u;y) \geq M - \varepsilon \|u\|_{C(\Gamma)}^2 \geq M - \varepsilon \|u\|_{\widetilde{C}^1(\Gamma)}^2.
$$
This proves condition (i). The same argument also applies to $\Phi_p(\cdot;\cdot)$. 

Let us now prove condition (ii) for both potentials. As for the first condition, the same proof works for both, and the constants will not depend on $p$. Thus, to avoid repetition, we will only prove for the potential $\Phi(\cdot;\cdot)$. Take any $r>0$ and let $\max\{\|u\|_{C^{0,\gamma}(\Gamma)},\|y\|_Y\} <r$ for $\alpha=1$, or  $\max\{\|u\|_{\widetilde{C}^1(\Gamma)},\|y\|_Y\} <r$. Observe that both cases imply
$$\|u\|_{C(\Gamma)} < r \quad\hbox{and}\quad \sum_{e\in\mathcal{E}} n_e < r.$$
Therefore,
\begin{align*}
    \Phi(u;y) &= \sum_{e\in\mathcal{E}} \left(\int_e \exp\{u_e(t)\}dt - \sum_{i=1}^{n_e} u(s_{i,e}) \right) \leq |\mathcal{E}| |\Gamma| \|u\|_{C(\Gamma)} + \|u\|_{C(\Gamma)} \sum_{e\in\mathcal{E}} n_e\\
    &\leq |\mathcal{E}| |\Gamma| r + r^2.
\end{align*}
This proves condition (ii). 

To conclude the proof, we need to obtain a bound for $|\Phi(u;y) - \Phi_p(u;y)|$. We start with the case $\alpha = 1$. 
First, let $f\in C^{0,\gamma}(\Gamma)$, and observe that by Lemma \ref{lem:holder_num_int_bound}, we obtain, by also using the assumption on the partition $\{q_{i,e}\}$, that
\begin{align*}
\sum_{e\in\mathcal{E}} \sum_{i=1}^{p_e} \int_{q_{i,e}}^{q_{i+1,e}} |f(t) - f(\tilde{s}_{i,e})| dt &\leq \frac{2}{1+\gamma} \|f\|_{C^{0,\gamma}(\Gamma)} \sum_{e\in\mathcal{E}} \sum_{i=1}^{p_e} |q_{i+1,e} - q_{i,e}|^{1+\gamma} \\
&\leq \frac{2}{1+\gamma} \|f\|_{C^{0,\gamma}(\Gamma)} \max_{i,e} |q_{i+1,e} - q_{i,e}|^{\gamma} \sum_{e\in\mathcal{E}} K_2 \sum_{i=1}^{p_e} p_e^{-1}\\
&\leq \frac{2 K_2|\mathcal{E}|}{1+\gamma} K_1^\gamma \|f\|_{C^{0,\gamma}(\Gamma)}\max_{e} p_e^{-\gamma} \leq C \|f\|_{C^{0,\gamma}(\Gamma)} \|p\|^{-\gamma},
\end{align*}
where the constant $C$ does not depend on $f$, nor on $p$. In a similar manner, for $\alpha = 2$ and using Lemma \ref{lem:C1_num_int_bound}, we obtain that there exists a constant $C$, not depending on $f$ or $p$, such that
$$\sum_{e\in\mathcal{E}} \sum_{i=1}^{p_e} \int_{q_{i,e}}^{q_{i+1,e}} |f(t) - f(\tilde{s}_{i,e})| dt \leq C \|f\|_{\widetilde{C}^1(\Gamma)} \|p\|^{-1}.$$

We can now move to bounding the difference of potentials. For $\alpha=1$,
\begin{align*}
|\Phi(u;y) - \Phi_p(u;y)| &\leq \sum_{e\in\mathcal{E}} \sum_{i=1}^{p_e} |\exp\{u_e(t)\} - \exp\{u_e(s_{i,e})\}| dt \leq C \|\exp\{u(\cdot)\}\|_{C^{0,\gamma}(\Gamma)} \|p\|^{-\gamma},
\end{align*}
and for $\alpha = 2$,
\begin{align*}
    |\Phi(u;y) - \Phi_p(u;y)| &\leq \sum_{e\in\mathcal{E}} \sum_{i=1}^{p_e} |\exp\{u_e(t)\} - \exp\{u_e(s_{i,e})\}| dt \leq C \|\exp\{u(\cdot)\}\|_{\widetilde{C}^1(\Gamma)} \|p\|^{-1}.
\end{align*}
By Theorem \ref{thm:cotter_stuart}, the result is proved if we show that for $\alpha=1$ for any $\varepsilon>0$, there exists $K = K(\varepsilon)>0$ such that
$$\|\exp\{u(\cdot)\}\|_{C^{0,\gamma}(\Gamma)} \leq K \exp\{\varepsilon \|u\|_{C^{0,\gamma}(\Gamma)}^2\},$$
and for $\alpha = 2$,
$$\|\exp\{u(\cdot)\}\|_{\widetilde{C}^1(\Gamma)} \leq K \exp\{\varepsilon \|u\|_{\widetilde{C}^1(\Gamma)}^2\}.$$

Let us first prove the case $\alpha = 1$. To this end, observe that $\exp(\cdot)$ is increasing, also that $x\leq \exp\{x\}$, and recall Young's inequality with $\varepsilon$, namely, $ab \leq \varepsilon a^2 + b^2/(4\varepsilon)$. Further, we also need \cite[Proposition 13]{GRAPH11} about composition of H\"older functions. Then,
\begin{align*}
    \|\exp\{u(\cdot)\}\|_{C^{0,\gamma}(\Gamma)} &\leq \|\exp\{u(\cdot)\}\|_{C(\Gamma)} + [\exp\{u(\cdot)\}]_{C^{0,\gamma}(\Gamma)}\\
    &\leq \exp\{\|u\|_{C(\Gamma)}\} + \exp\{\|u\|_{C(\Gamma)}\} \left(\|u\|_{C^{0,\gamma}(\Gamma)} \right)^\gamma\\
    &\leq \exp\{\|u\|_{C(\Gamma)}\} + \exp\{\|u\|_{C(\Gamma)} + \gamma \|u\|_{C^{0,\gamma}(\Gamma)}\}\\
    &\leq \exp\{\|u\|_{C^{0,\gamma}(\Gamma)}\} + \exp\{(1+\gamma)\|u\|_{C^{0,\gamma}(\Gamma)}\}\\
    &\leq 2 \exp\{(1+\gamma)\|u\|_{C^{0,\gamma}(\Gamma)}\}\leq K \exp\{\varepsilon \|u\|_{C^{0,\gamma}(\Gamma)}^2\},
\end{align*}
where $K = K(\varepsilon) = 2\exp\{(1+\gamma)^2/(4\varepsilon)\}$. This concludes the proof for $\alpha = 1$. Now, for $\alpha=2$ the proof is similar:
\begin{align*}
    \|\exp\{u(\cdot)\}\|_{\widetilde{C}^1(\Gamma)} &\leq \|\exp\{u(\cdot)\}\|_{C(\Gamma)} + \sum_{e\in\mathcal{E}} \|\exp\{u_e(\cdot)\}u_e'(\cdot)\|_{C(e)}\\
    &\leq \exp\{\|u\|_{C(\Gamma)}\} + \exp\{\|u\|_{C(\Gamma)}\} \sum_{e\in\mathcal{E}} \|u'\|_{C(e)}\\
    &\leq \exp\{\|u\|_{C(\Gamma)}\} + |\mathcal{E}| \exp\{\|u\|_{C(\Gamma)}\} \|u\|_{\widetilde{C}^1(\Gamma)}\\    
    &\leq \exp\{\|u\|_{C(\Gamma)}\} +  |\mathcal{E}| \exp\{\|u\|_{C(\Gamma)} + \|u\|_{\widetilde{C}^1(\Gamma)}\}\\
    &\leq \exp\{\|u\|_{\widetilde{C}^1(\Gamma)}\} + |\mathcal{E}|\exp\{2\|u\|_{\widetilde{C}^1(\Gamma)}\}\\
    &\leq (1+|\mathcal{E}|) \exp\{2\|u\|_{\widetilde{C}^1(\Gamma)}\}\leq K \exp\{\varepsilon \|u\|_{\widetilde{C}^1(\Gamma)}^2\},
\end{align*}
where, here, $K = K(\varepsilon) = (1+|\mathcal{E}|)\exp\{1/\varepsilon\}$. This concludes the proof for $\alpha = 2$.
\end{proof}

In view of Propositions \ref{prop:regularity_results} and \ref{prop:tau_and_u}, we obtain the following corollary:

\begin{corollary}
    \label{cor:convergence_tau_cases}
    Under the setting of Theorem \ref{thm:main_theorem}, suppose that \(\tau\) is either constant or given by \(\tau(s) = \sigma^{-1} \sigma_\kappa(s)\). Then, for \(\alpha = 1\), the convergence rate is given by
    $$
    d_{\text{Hell}}(\mu, \mu_p) \leq C \|p\|^{-\gamma}, \quad \text{for any } 0 < \gamma < \frac{1}{2}.
    $$
    For \(\alpha = 2\), we obtain the convergence rate
    $$
    d_{\text{Hell}}(\mu, \mu_p) \leq C \|p\|^{-1}.
    $$
\end{corollary}

\begin{proof}[Proof of Theorem \ref{thm:hellinger_convergence}]
    This follows directly from Corollary \ref{cor:convergence_tau_cases}.
\end{proof}

\bibliographystyle{abbrvnat}
\bibliography{references}

\end{document}